\documentclass[11pt]{article}

\usepackage[latin1]{inputenc}
\usepackage[T1]{fontenc}
\usepackage[english]{babel}
\usepackage{amsfonts}
\usepackage{amssymb}
\usepackage{amsmath}
\usepackage{amsthm}
\usepackage{graphicx}

\def\be{\begin{equation}}
\def\ee{\end{equation}}
\def\bc{\begin{center}}
\def\ec{\end{center}}

 \newtheorem{thm}{Theorem}[section]

 \newtheorem{prop}[thm]{Proposition}
 \theoremstyle{definition}
 
 \theoremstyle{remark}
 \newtheorem{rem}[thm]{Remark}
 \newtheorem*{ex}{Example}
 \numberwithin{equation}{section}

\def\be{\begin{equation}}
\def\ee{\end{equation}}
\def\bc{\begin{center}}
\def\ec{\end{center}}

\begin{document}

\title{Multi-species mean-field spin-glasses.\\
 Rigorous results}

\author{Adriano Barra\footnote{Dipartimento di Fisica, Sapienza Universit\`a di Roma}, \ Pierluigi Contucci \footnote{Dipartimento di Matematica, Universit\`a di Bologna}, \ Emanuele Mingione \footnote{Dipartimento di Matematica, Universit\`a di Bologna} \ and Daniele Tantari \footnote{Dipartimento di Matematica, Sapienza Universit\`a di Roma}}

\maketitle


\begin{abstract}
We study a multi-species spin glass system where the density of each species is kept fixed at increasing
volumes. The model reduces to the Sherrington-Kirkpatrick one for the single species case. The existence
of the thermodynamic limit is proved for all densities values under a convexity condition on the interaction.
The thermodynamic properties of the model are investigated and the annealed, the replica symmetric and the replica symmetry breaking bounds are proved using Guerra's scheme.
The annealed approximation is proved to be exact
under a high temperature condition. We show that the replica symmetric solution has negative entropy at low temperatures. We study the properties of a suitably defined replica symmetry breaking solution and we optimise it within a novel {\em ziggurat ansatz}. The generalised order parameter is described by a Parisi-like partial differential equation.
\end{abstract}


\maketitle

\section{Introduction and main results}

In this paper we introduce and study a multi-species mean field spin glass model\footnote{see also other works on extensions to multipartite of standard mean-field spin glasses \cite{BGG2,Bovier,talagrand1,talagrand2} or ferromagnets \cite{cg,fc,fu}.}, i.e. a system composed by spins belonging to a finite number of
different species, where spin couples interact through a centered gaussian variable whose variance
depends only on the species they belong to.
\newline
We rigorously prove bounds for the model pressure
and control their properties in a region of convexity defined in terms of the variance matrix
of the interactions, namely when the interactions within a group dominate the inter-groups interactions. In particular we prove the existence of the thermodynamic limit of the model, the correctness of the annealed solution in the high temperature limit and we perform an extension of the broken replica symmetry bounds developed by Guerra for the SK model \cite{broken} suitably adapted to this multi-species extension. In particular we formulate  a novel {\em ziggurat} ansatz for the broken replica phase which acts as a rigorous upper bound for the free energy.
Before starting to explore the model, it may be worth stressing that, while we introduced the model and developed the Guerra's bound for the free energy, Panchenko proved also the reverse bound \cite{panchenkomulti} hence the two papers as a whole confirm the correctness of our ansatz.

%
More specifically, the system is composed by a finite number $S$ of species indexed by $s\in\mathcal{S}$. For each $s\in\mathcal{S}$ consider the set $\Lambda^{(s)}_N\subset\mathbb{Z}$ such that
\begin{gather}
\label{par1}
\Lambda^{(s)}_N\bigcap\Lambda^{(p)}_N=\emptyset,\ \ \ \forall s\neq p\in\mathcal{S}\\
\label{par2}
|\Lambda^{(s)}_N|=N^{(s)},\\
\label{par3}
N=\sum_{s\in\mathcal{S}}N^{(s)}.
\end{gather}
\noindent
Consider a disordered model defined by a collection $(\sigma^{(s)})_{s\in\mathcal{S}}$ of Ising variables, meaning that $\sigma^{(s)}_i=\pm 1$ for each $\forall s\in\mathcal{S}, i\in \Lambda^{(s)}_N$. We denote by $\Sigma_{N}$ the family of possible spin configurations $\sigma=\{\sigma_i^{(s)}\}_{s\in\mathcal{S},i\in \Lambda^{(s)}_N}$, then  we have that $\mid\Sigma_{N}\mid=2^N$. In the sequel the following standard definitions of spin glass theory will be used.\smallskip

For every $N\in \mathbb{N}$, let $\{H_{N}(\sigma)\}_{\sigma\in\Sigma_{N}}$  be a family of $2^N$  gaussian r.v. which represent the Hamiltonian of the system, defined by
\be\label{hami}
H_{N}(\sigma):=-\frac{1}{\sqrt{N}}\sum_{s,p\in\mathcal{S}} \sum_{i\in\Lambda^{(s)}_N}\sum_{j\in\Lambda^{(p)}_N} J^{(sp)}_{ij}\sigma^{(s)}_i\sigma^{(p)}_j,
\ee
where the $J$'s are \textsl{Gaussian i.i.d.} r.v. such that for every $s,p,i,j$ we have that
\be\label{mean}
\mathbb{E}(J^{(sp)}_{ij})=0,
\ee
and
\be\label{varia}
\mathbb{E}(J^{(sp)}_{ij}J^{(s'p')}_{i'j'})=\delta_{ss'}\delta_{pp'}\delta_{ii'}\delta_{jj'}\Delta^2_{sp},
\ee
with $\Delta^2_{sp}=\Delta^2_{ps}$.\smallskip

The covariance matrix of the system is
\be\label{Covamodel}
C_N(\sigma,\tau):=\mathbb{E}(H_{N}(\sigma)H_{N}(\tau))
\ee
and, thanks to (\ref{mean}) and(\ref{varia}), a simple computation shows that
\be\label{covamodel2}
C_N(\sigma,\tau)=\frac{1}{N}\sum_{s,p\in\mathcal{S}}\Delta^2_{sp} \Big(\sum_{i\in\Lambda^{(s)}_N}\sigma^{(s)}_i\tau^{(s)}_i\Big)  \Big(\sum_{j\in\Lambda^{(p)}_N}\sigma^{(p)}_j\tau^{(p)}_j\Big).
\ee
To show explicitly the dependence trough the choice of the various sizes we can define for every $s\in\mathcal{S}$ the relative density
\be\label{reles}
\alpha^{(s)}_N:=\frac{N^{(s)}}{N},
\ee
and the relative overlap
\be\label{relover}
q^{(s)}_N(\sigma,\tau)=\frac{1}{N^{(s)}}\sum_{i\in\Lambda^{(s)}_N}\sigma^{(s)}_i\tau^{(s)}_i,
\ee
then the covariance matrix can be write in the form
\be\label{cova2}
C_N(\sigma,\tau)=N\sum_{s\in\mathcal{S}}\sum_{p\in\mathcal{S}}\Delta^2_{sp}\alpha_N^{(s)}\alpha^{(p)}_N q^{(s)}_N(\sigma,\tau) q^{(p)}_N(\sigma,\tau).
\ee
We can introduce in a natural way a vector notation ( see Appendix \ref{notation} for a detailed explanation), and rewrite (\ref{cova2}) as
\be\label{quad}
C_N(\sigma,\tau)=N\Big(\mathbf{q}_N,\mathbf{\Delta} \mathbf{q}_N\Big),
\ee
where
\be\label{quadvec}
\mathbf{q}_N=\Big(q^{(s)}_N(\sigma,\tau)\Big)_{s\in\mathcal{S}}
\ee
\begin{ex}
For example, in the case of two species, namely $\mathcal{S}=\{a,b\}$, the covariance matrix  $\mathbf{q}_N$ is a 2-dimensional vector and $\mathbf{\Delta}$ is a $2\times2$ matrix defined by the entries

\[\begin{bmatrix}
\alpha_N^{(a)}\alpha_N^{(a)}\Delta^2_{aa}&\alpha_N^{(a)}\alpha_N^{(b)}\Delta^2_{ab}\\
\alpha_N^{(a)}\alpha_N^{(b)}\Delta^2_{ab}&\alpha_N^{(b)}\alpha_N^{(b)}\Delta^2_{bb}\\
\end{bmatrix}\]
\end{ex}
\medskip

The random partition function extends standard partition function for disordered systems and reads off as
\be\label{partition}
Z_N:=\sum_{\sigma}a_N(\sigma,\mathbf{h})e^{- H_{N}(\sigma)},
\ee
where
\be\label{weights}
a_N(\sigma,\mathbf{h}):=\exp\Big(\sum_{s\in\mathcal{A}}h^{(s)} \sum_{i\in\Lambda^{(s)}_N}\sigma^{(s)}\Big),
\ee
and $\mathbf{h}:=(h^{(s)})_{s\in\mathcal{A}}$ is a vector which represents an external (magnetic in the physical literature) field acting in each party separately.\smallskip

\begin{rem} Notice that, to lighten the notation, in the \textsl{l.h.s.} of (\ref{partition}), and in the rest of the paper, we do not  write explicitly the dependence on $\mathbf{h}$. With the same aim, the physical inverse temperature $\beta$, which appears in the standard definition of the partition function, in our case is set equal to $1$ with no loss of generality as it can be recovered in every moment simply by properly rescaling the interactions parameters.
\end{rem}
\noindent
\medskip

The main thermodynamic observable is the quenched  pressure density, defined as
\be\label{qpres}
p_N:=\frac{1}{N}P_N:=\frac{1}{N}\mathbb{E}\log Z_N
\ee
It's useful also to introduce a basic object in spin glass theory, the so-called quenched Gibbs measure.
\smallskip

We consider $n$ copies of the configuration space, denoted by  $\sigma^{1},\ldots, \sigma^{n}$ and,
for every bounded function $f:(\sigma^{1},\ldots, \sigma^{n})\rightarrow \mathbb{R}$, we call the random $n$-Gibbs state the following r.v.

\be\label{rgistate}
\Omega_{N}(f):=\sum_{\sigma^{1},\ldots, \sigma^{n}}
f(\sigma^{1},\ldots, \sigma^{n})G_{N}(\sigma^{1})\ldots G_{N}(\sigma^{n})
\ee
where
\be\label{gimeas}
G_{N}(\sigma):=\frac{a_N(\sigma,\mathbf{h})e^{-H_{N}(\sigma)}}{Z_N}
\ee
is called Gibbs measure.
\smallskip
Lastly we define the quenched Gibbs measure as
\be\label{qeqstate}
\langle f\rangle_N:= \mathbb{E}\Omega_{N}(f).
\ee

First at all we prove, under a suitable condition, the existence of the thermodynamic limit for the quenched pressure density (\ref{qpres}) when the species densities are kept constant, i.e. the limit of ${N \to  \infty}$ is defined such that $\forall s\in\mathcal{S}$  the quantity $\alpha^{(s)}_N=\frac{N^{(s)}}{N}=\alpha^{(s)}$ is independent of $N$. More precisely, we prove the following
\begin{thm}\label{therlimite}
If the matrix $\mathbf{\Delta}$ is positive semi-definite, then
$$\lim_{N \to  \infty}p_N=\sup_N p_N,$$
where the limit is taken at fixed densities.
\end{thm}
Notice that, since the relative densities are kept constants, the condition that $\mathbf{\Delta}$ is positive semi-definite is independent of the $\alpha$'s.\\

In the single species case the model reduces to the Sherringhton-Kirpatrick model. Previous work \cite{broken}, by using the interpolation method, show that the celebrated Parisi's solution is an upper bound for pressure. The interpolation ( usually called RSB interpolation) is defined through a non-decreasing, piecewise constant function $x(q):[0,1] \longrightarrow [0,1]$ which represents the cumulative distribution of the overlap (see \cite{talbook2}) and the order parameter of the model. One of the key points of the proof, is that the function $x(q)$ intrinsically defines a non decreasing sequence $(m_l)_{l=0,..., K+1}$, thus enabling the control of the sign of the derivative of the interpolating functional. Following the same approach, in the multi-species case we define the order parameter as a suitable piecewise constant, right continuous, function
\be\label{orderparadef}
x(\mathbf{u}):[0,1]^{S} \rightarrow [0,1].
\ee
Let us show the explicit construction. Let $K$ be an integer and consider a sequence of points in $\Gamma\in [0,1]^{S}$, with
\be\label{qtrialrsb}
\Gamma:=(\mathbf{q}_l)_{l=0,\ldots,K}=\Big(q_l^{(s)}\Big)_{s\in \mathcal{S}, l=0,\ldots,K}
\ee
such that for each $s\in \mathcal{S}$, we have
$$0=q_0^{(s)}\leq q_1^{(s)}\leq\cdots\leq q_{K-1}^{(s)}\leq q_{K}^{(s)}=1.$$
Roughly speaking $\Gamma$ defines a path with $K$ steps in $[0,1]^S$ which is non decreasing in each direction.
\\
We consider a sequence $(m_0,..., m_{K+1})$ such that $0=m_0\leq m_1\leq \cdots \leq m_{K}\leq m_{K+1}=1$.
\\
If we denote by $\theta(\cdot)$  the right continuous Heaviside  function, we define the functional order parameter as
\be\label{orderpara}
x(\mathbf{u}):=\sum^{K}_{l=0}(m_{l+1}-m_l) \prod_{s\in \mathcal{S}}\theta(u^{(s)}-q_l^{(s)})
\ee
where $\mathbf{u}=(u^{(s)})_{s\in\mathcal{S}}$ is vector in $[0,1]^S$.
\\
The function $x$ defines an $S$-dimensional shape, that in the case of $S=2$, looks like a $ziggurat$\footnote{Ziggurats are pyramid-like structures found in the ancient Mesopotamian valley and western Iranian plateau.} and clearly, in the analogy of the single species case, it must be represent the joint cumulative distribution of the relatives overlaps.
\\
It's  useful to introduce, for each $s\in \mathcal{S}$ the canonical projection operator $\mathcal{P}_s$ in $\mathbb{R}^S$, and for $l=0,\ldots,K$, the following quantity:
\be\label{trialbig}
Q_l^{(s)}=\frac{2}{\alpha^{(s)}}\mathcal{P}_s\Big(\mathbf{\Delta}\mathbf{q}_l\Big)
\ee
where $\mathcal{P}_s$ is the canonical projection operator  in $\mathbb{R}^S$ define in (\ref{projection}).
\\
Notice that, for each $s\in \mathcal{S}$, the sequence $(Q_l^{(s)})_{l=0,\ldots K}$ is non-decreasing.
To complete the picture we need to introduce a transformed order parameter
\be\label{modorderpara}
x_{\Delta}(\mathbf{u}):=\sum^{K}_{l=0}(m_{l+1}-m_l)\prod_{s\in\mathcal{S}} \theta(u^{(s)}-Q_l^{(s)})
\ee
defined for $ \mathbf{u}\in \times_{s\in \mathcal{S}}[0,Q_K^{(s)}]$.\medskip

We define the trial RSB pressure as
\be\label{parisiRSB}
p_{RSB}(x):=\log2+\sum_{s\in\mathcal{S}}\alpha{(s)}f^{(s)}(0,h^{(s)})
-\frac{1}{2}\int_{\widetilde{\Gamma}}\ \ x(\mathbf{u})\ \ \nabla_{\mathbf{u}}\Big(\mathbf{u},\mathbf{\Delta}\mathbf{u}\Big)\cdot d\mathbf{u}
\ee
where, for each $s\in \mathcal{S}$, $f^{(s)}(u^{(s)},y)$ is the solution of the following Parisi's PDE
\begin{equation}
\label{antipara} \frac{\partial f^{(s)}}{\partial u^{(s)}} +
{1\over2}\frac{\partial^2f^{(s)}}{\partial y^2}+{1\over2}x_{\mathbf{\Delta}}(u^{(s)})\Big(\frac{\partial f^{(s)}}{\partial y}\Big)^2=0
\end{equation}
where $x_{\mathbf{\Delta}}(u^{(s)})$ is the marginal value of the transformed order parameter and the boundary condition is
\be\label{parisibou}
f^{(s)}(Q_K^{(s)},y)=\log\cosh(y).
\ee
The integral in (\ref{parisiRSB}) is a line integral on an arbitrary path $\widetilde{\Gamma}$ in the plan $\mathbf{u}$, starting from $\mathbf{0}$ and ending in $\mathbf{1}$, such that all the points $(\mathbf{q}_l)_{l=0,\ldots,K}$ belong to $\Gamma$, in other words $\Gamma\subset\widetilde{\Gamma}$.\smallskip

The main result of the paper is the next
\begin{thm}\label{thrsb}
The following sum rule holds
\be\label{sumrulersb}
p_N=p_{RSB}(x)-\frac{1}{2}\sum^K_{l=0}(m_{l+1}-m_l)\int^1_0dt\Big\langle(\mathbf{q}_N-\mathbf{q}_{l}),\mathbf{\Delta} (\mathbf{q}_N-\mathbf{q}_{l})\Big\rangle_{N,l,t}.
\ee
Moreover if the matrix $\mathbf{\Delta}$ is positive semi-definite we have the following bound
$$p_N\leq p_{RSB}(x),$$
and the optimization gives
$$
p_N\leq \inf_{x} p_{RSB}(x).
$$
\end{thm}
We stress the fact that the {\em ziggurat} ansatz  forces the joint overlap distribution to have a very special structure.
In a recent paper \cite{panchenkomulti}, which appeared
after this work, the author proved that the previous bound is exact, showing that this structure in fact encodes all the information of the model.
\\
All the results are given under the hypothesis that the mutual interaction does not exceed a threshold where the strength of the off-diagonal terms prevails on the inter-party interactions. The latter is intrinsically different because the model approaches the Hopfield model for neural network \cite{BBCS,BGG1,BGGT}, on which we plan to report soon.
\\
To frame properly the main result within a clear scenario for the multi-species mean field spin glass, the paper is organized as follows: In section 2 we prove Theorem \ref{therlimite}, i.e. the existence of thermodynamic limit.
Section 3 studies the annealed region with the second moment method. Section 4 considers the replica symmetric bound and shows that at low temperatures it has a negative entropy. Finally in section 5 we give an extensive proof of  Theorem \ref{thrsb}.

\section{The thermodynamical limit at fixed densities}

In this section we prove Theorem \ref{therlimite}. The strategy of the proof follows classical Guerra-Toninelli arguments. Let us consider two non interacting and  $i.i.d.$ copies of the original system defined by the Hamiltonian (\ref{hami}) of sizes respectively $N_1,\ N_2$. Clearly this implies that we have to consider $\forall s\in\mathcal{S}$ the relative subsets $\Lambda^{(s)}_{N_1},\Lambda^{(s)}_{N_2}$ defined by the equations (\ref{par1}), (\ref{par2}), (\ref{par3}) and such that
\begin{gather*}
\Lambda^{(s)}_{N_1}\cup\Lambda^{(s)}_{N_2}=\Lambda^{(s)}_N,\\
\Lambda^{(s)}_{N_1}\cap\Lambda^{(s)}_{N_2}=\emptyset, \\
|\Lambda^{(s)}_{N_1}|=N_1^{(s)},\\
|\Lambda^{(s)}_{N_2}|=N_2^{(s)},\\
N_1^{(s)}+N_2^{(s)}=N^{(s)}.
\end{gather*}
More explicitly, we can define,  $\forall s\in\mathcal{S}$, the following
\begin{gather}
\label{labelling}
\Lambda^{(s)}_{N}=\{1,\ldots, N^{(s)}\},\\
\label{labelling1}
\Lambda^{(s)}_{N_1}=\{1,\ldots, N_1^{(s)}\},\\
\label{labelling2}
\Lambda^{(s)}_{N_2}=\{N_1^{(s)}+1,\ldots, N^{(s)}\}.
\end{gather}
Consider the following interpolating Hamiltonian
\be\label{haminter}
H_N(\sigma,t)=\sqrt{t}H_N(\sigma)+\sqrt{1-t}\Big(H_{N_1}(\sigma)+H_{N_2}(\sigma)\Big)
\ee
where
\be\label{haminter1}
H_{N_1}(\sigma)=-\frac{1}{\sqrt{N_1}}\sum_{s,p\in\mathcal{S}} \sum_{i\in\Lambda^{(s)}_{N_1}}\sum_{j\in\Lambda^{(p)}_{N_1}} J'^{(sp)}_{ij}\sigma^{(s)}_i\sigma^{(p)}_j,
\ee
\be\label{haminter2}
H_{N_2}(\sigma)=-\frac{1}{\sqrt{N_2}}\sum_{s,p\in\mathcal{S}} \sum_{i\in\Lambda^{(s)}_{N_2}}\sum_{j\in\Lambda^{(p)}_{N_2}} J''^{(sp)}_{ij}\sigma^{(s)}_i\sigma^{(p)}_j,
\ee
and where $J'^{(sp)}_{ij}$ and $J''^{(sp)}_{ij}$ are $i.i.d.$ of $J^{(sp)}_{ij}$.
\\
As usual we consider the interpolating  pressure
\be\label{preinter2}
P_N(t)=\mathbb{E}\log Z_N(t)=\mathbb{E}\log\sum_{\sigma}a_N(\sigma,\mathbf{h})e^{-H_{N}(\sigma,t)},
\ee
whose boundaries values are
\begin{eqnarray}\label{boundary}
P_N(1)&\equiv& P_N, \\
P_N(0)&\equiv& P_{N_1}+P_{N_2},
\end{eqnarray}
since $\Sigma_{N}=\Sigma_{N_1}\cup \Sigma_{N_2}$ and $\Sigma_{N_1}\cap \Sigma_{N_2}=\emptyset$.

\begin{prop}\label{derivativeTL}
The $t$-derivative of the interpolating pressure is
$$\frac{\partial}{\partial t}P_N(t)=-\frac{N}{2}\mathbb{E}\Omega_{N,t}
\Big(Q_N\Big),$$
where
\be\label{deltaq}
Q_N(\sigma,\tau):=\Big(\mathbf{q}_N, \mathbf{\Delta} \mathbf{q}_N\Big)-\frac{N_1}{N}\Big(\mathbf{q}_{N_1}, \mathbf{\Delta} \mathbf{q}_{N_1}\Big)-\frac{N_2}{N}\Big(\mathbf{q}_{N_2}, \mathbf{\Delta} \mathbf{q}_{N_2}\Big),
\ee
and the vectors $\mathbf{q}_{N_1},\mathbf{q}_{N_2}$ are defined as in \eqref{quadvec}.
\end{prop}
 \begin{proof}The computation of the $t$-derivative works essentially in the same way exploited  in Proposition \ref{gauinter} with the following identifications:
$$i\rightarrow\sigma, \ \ a_i\rightarrow a_N(\sigma,\mathbf{h}), \ \ U_i\rightarrow H_{N}(\sigma), \ \
\widetilde{U}_i\rightarrow H_{N_1}(\sigma)+H_{N_2}(\sigma)$$
The key ingredient is that the diagonal term vanishes by the condition $N=N_1+N_2$.
\end{proof}
Combining the Fundamental Theorem of Calculus and the previous proposition we have that
\be\label{haminter12}
P_N-P_{N_1}-P_{N_2}=-\frac{N}{2}\int^1_0dt\mathbb{E}\Omega_{N,t}
\Big(Q_N\Big).
\ee
To finish the proof is sufficient to show that
\begin{prop}\label{convcond}
If the matrix $\mathbf{\Delta}$ is positive semi-definite, then
\be\label{estimate}
Q_N(\sigma,\tau)\leq 0
\ee
for every $\sigma,\tau$ and $N$.
\end{prop}
\begin{proof}
First at all, we write some fundamental relations.
\newline
By definitions (\ref{relover}), (\ref{labelling}), (\ref{labelling1}), (\ref{labelling2}) we have that $\forall s\in\mathcal{S}$ the following hold
$$N^{(s)}q^{(s)}_{N}(\sigma,\tau)=\sum_{i=1}^{N^{(s)}}\sigma^{(s)}_i\tau^{(s)}_i=
\sum_{i=1}^{N_1^{(s)}}\sigma^{(s)}_i\tau^{(s)}_i+\sum_{N^{(s)}+1}^{N^{(s)}}\sigma^{(s)}_i\tau^{(s)}_i
$$
then
$$
q^{(s)}_{N}(\sigma,\tau)=\frac{N_1^{(s)}}{N^{(s)}}q^{(s)}_{N_1}(\sigma,\tau)+\frac{N_2^{(s)}}{N^{(s)}}q^{(s)}_{N_2}(\sigma,\tau).$$
Now we observe that the condition of fixed relatives densities  implies that
$$\frac{N_1^{(s)}}{N^{(s)}}=\frac{N_1^{(s)}}{N_1}\frac{N}{N^{(s)}}\frac{N_1}{N}=
\frac{\alpha^{(s)}}{\alpha^{(s)}}\frac{N_1}{N}=\frac{N_1}{N},$$
and in a similar fashion
$$\frac{N_2^{(s)}}{N^{(s)}}=\frac{N_2}{N},$$
then $\forall s\in\mathcal{S}$ the following holds
\be\label{qsumrule}
q^{(s)}_{N}(\sigma,\tau)=\frac{N_1}{N}q^{(s)}_{N_1}(\sigma,\tau)+\frac{N_2}{N}q^{(s)}_{N_2}(\sigma,\tau).
\ee
In vector notation we can write
\be\label{qsumrulev}
 \mathbf{q}_N=\frac{N_1}{N}\mathbf{q}_{N_1}+\frac{N_2}{N}\mathbf{q}_{N_2}.
\ee
It is easy to see that if $\mathbf{\Delta}$ is a positive semi-definite, real, symmetric matrix, hence the function
$$ \mathbf{x}\rightarrow \Big(\mathbf{x},\mathbf{\Delta} \mathbf{x}\Big)$$
defined for $ \mathbf{x} \in \mathbb{R}^S$ is convex and the conclusion follows straightforwardly from the relation (\ref{qsumrulev}).
\end{proof}

The last proposition, combined with equation (\ref{haminter}) gives immediately the superaddivity property of the pressure. As a consequence, since the quenched pressure density is bounded from the annealed one (see the next section), then by Fakete's lemma we get the statement of the theorem.

\section{The annealed bound}

As a first analysis we can study the annealed approximation for the pressure and investigate in which case it is exact. Using Jensen inequality and the concavity of the function $x\to\log (x)$ we define the annealed approximation as a bound, i.e.
\be
p_N=\frac 1 N \mathbb{E} \log Z_N\leq \frac 1 N \log \mathbb{E}Z_N =p^{A}_N.
\ee
We can easily write $p_N^A$ as
\begin{eqnarray}
p^A_N&=&\frac 1 N \log \sum_{\sigma}\mathbb{E}e^{-H_N(\sigma)}=
\frac 1 N \log \sum_{\sigma}e^{\frac 1 2 C_N(\sigma,\sigma)}=
\frac 1 N \log \sum_{\sigma}e^{\frac N 2 (\boldsymbol{1,\Delta 1})}\nonumber\\
&=&\log 2+\frac 1 2(\boldsymbol{1,\Delta 1}).
\end{eqnarray}
We define the ergodic regime as the region of the phase space in which
\be
\lim_{N\to\infty}\frac 1 N\mathbb{E}\log Z_N=\lim_{N\to\infty}\frac 1 N\log \mathbb{E} Z_N=p^A=\log 2+\frac 1 2(\mathbf{1,\Delta 1}).
\ee
For this purpose we can use a Borel-Cantelli argument to investigate the second moment, hence checking when
\be
\frac{\mathbb{E}(Z_N^2)}{\mathbb{E}^2(Z_N)}\leq C<\infty
\ee
for some constant $C\in\mathbb{R}$, uniformly in $N$. Since
\begin{eqnarray}
\mathbb{E}(Z_N^2)&=& \mathbb{E} \sum_{\sigma,\tau}e^{-H_N(\sigma)-H_N(\tau)}=\sum_{\sigma,\tau}e^{\frac 1 2 \mathbb{E}(H_N(\sigma)+H_N(\tau))^2}\\ \nonumber
&=&\sum_{\sigma,\tau}e^{N \left( (\boldsymbol{1,\Delta 1})+(\mathbf{q}_N,\boldsymbol{\Delta} \mathbf{q}_N)  \right) }= \mathbb{E}^2(Z_N) 2^{-2N}\sum_{\sigma,\tau}e^{N(\mathbf{q}_N,\boldsymbol{\Delta} \mathbf{q}_N)}
\end{eqnarray}
and using the gauge transformation $\tau^{(s)}_i\to\sigma^{(s)}_i\tau^{(s)}_i$,
\be
\frac{\mathbb{E}(Z_N^2)}{\mathbb{E}^2(Z_N)}=2^{-2N}\sum_{\sigma,\tau}e^{N(\mathbf{m}_N(\tau),\boldsymbol{\Delta} \mathbf{m}_N(\tau))}=
2^{-N}\sum_{\tau}e^{N(\mathbf{m}_N(\tau),\boldsymbol{\Delta} \mathbf{m}_N(\tau))},
\ee
where we define $\mathbf{m}_N(\tau)=\Big(m^{(s)}_N(\tau)\Big)_{s\in\mathcal{S}}$, with $m^{(s)}_N(\tau)=\frac 1 {N^{(s)}}\sum_{i=1}^{N^{(s)}}\tau^{(s)}_i$.
If $\det \boldsymbol{\Delta}>0$ we can linearize the quadratic form with a gaussian integration
\begin{eqnarray}
\frac{\mathbb{E}(Z_N^2)}{\mathbb{E}^2(Z_N)}&=&\frac {2^{-N}} {\sqrt{\det \boldsymbol{\Delta}}}\int\frac{d\mathbf{z}}{2\pi}e^{-\frac 1 2 (\mathbf{z},\boldsymbol{\Delta}^{-1}\mathbf{z})}\sum_{\tau}e^{\sqrt{2N}(\mathbf{m}_N(\tau),\mathbf{z})}\nonumber\\
&=&\frac 1 {\sqrt{\det \boldsymbol{\Delta}}}\int\frac{d\mathbf{z}}{2\pi}e^{-\frac 1 2 (\mathbf{z},\boldsymbol{\Delta}^{-1}\mathbf{z})}\prod_{s\in\mathcal{A}}\cosh^{N^{(s)}}\left( \frac{\sqrt{2N}}{N^{(s)}}z^{(s)} \right)\nonumber\\
&=&\frac 1 {\sqrt{\det (\boldsymbol{\Delta})}}\int\frac{d\mathbf{z}}{2\pi}e^{-\frac 1 2 (\mathbf{z},\boldsymbol{\Delta}^{-1}\mathbf{z})}e^{\sum_{s\in\mathcal{S}}N^{(s)}\log\cosh\left( \frac{\sqrt{2N}}{N^{(s)}}z^{(s)} \right)}
\end{eqnarray}
and, using the inequality $\log \cosh (x)\leq \frac{x^2}{2}$, we obtain
\be
\frac{\mathbb{E}(Z_N^2)}{\mathbb{E}^2(Z_N)}\leq \frac 1 {\sqrt{\det (\boldsymbol{\Delta})}}\int\frac{d\mathbf{z}}{2\pi}e^{-\frac 1 2 (\mathbf{z},\boldsymbol{\hat{\Delta}}\mathbf{z})},
\ee
where we have defined
\be
\boldsymbol{\hat{\Delta}}=\boldsymbol{\Delta}^{-1}-2\boldsymbol{\alpha}^{-1}
\ee
and the diagonal matrix $\boldsymbol{\alpha}=\operatorname{diag}(\{\alpha^{(s)}\}_{s\in\mathcal{S}}) $. Thus we have just proved the following
\begin{thm}\label{annealed region}
In the convex region, defined as $\det \boldsymbol{\Delta}>0$, as soon as $\boldsymbol{\hat{\Delta}}$ is positively defined, the pressure of the model does coincide with the annealed approximation, i.e.
\be
p=\lim_{N\to\infty}\frac 1 N\mathbb{E}\log Z_N=\lim_{N\to\infty}\frac 1 N\log \mathbb{E} Z_N=p^A=\log 2+\frac 1 2(\mathbf{1,\Delta 1}).
\ee
\end{thm}
\begin{rem}
\noindent Note that such a region does exist and can be viewed as an high temperature region. The two regions $\det\boldsymbol{\Delta}>0$ and $\boldsymbol{\hat{\Delta}}>0$ have a non-zero measure intersection, because, while the first is a condition on the relative size of the covariances, the latter is related to their absolute amplitude. Indeed once fixed $\boldsymbol{\alpha}$ and $\boldsymbol{\Delta}$ satisfying $\det \boldsymbol{\Delta} >0$, we can rescale all the covariances with a parameter $\beta$, which play the role of the inverse temperature of the system,  i.e. $\Delta_{ss'}\to \beta \Delta_{ss'}$, $\forall s,s'\in\mathcal{S}$, leaving the relative sizes unaltered and the condition $\det \boldsymbol{\Delta}>0$ is still satisfied, such that $\boldsymbol{\hat{\Delta}}\to \beta^{-S}\boldsymbol{\Delta}^{-1}-2\boldsymbol{\alpha}^{-1}$ is positively defined for $\beta$ small enough \footnote{Since $\boldsymbol{\Delta}$ is positively defined then also $\boldsymbol{\Delta}^{-1}$. Defining $a=\max_s \alpha^{(s)}$ and $\rho$ the smallest eigenvalue of $\boldsymbol{\Delta}^{-1}$, then, for any non-null vector $z$, $(z,\boldsymbol{\hat{\Delta}} z)\geq (\beta^{-S}\rho-a)(z,z)>0$ if $\beta^S<\rho/a$.}.
\end{rem}

\section{The Replica Symmetric bound}

In this section we specialize Theorem \ref{thrsb} in the simplest case  to obtain the so called RS bound.
\\
The underlying idea is to compare the overlap vector (\ref{quadvec}) with a trial vector,
\be\label{qtrial}
\mathbf{q}_{trial}:=\Big(q^{(s)}\Big)_{s\in \mathcal{S}}.
\ee

We define the trial replica symmetric solution as
\be\label{rssolution}
p_{RS}(\mathbf{q}_{trial}):=\log2+\sum_{s\in\mathcal{S}}\alpha^{(s)}p^{(s)}(\mathbf{q}_{trial})
+\frac{1}{2}\Big((\mathbf{1}-\mathbf{q}_{trial}),\mathbf{\Delta}(\mathbf{1}-\mathbf{q}_{trial})\Big),
\ee
where
\be\label{scomprs}
p^{(s)}(\mathbf{q}_{trial}):=\int d\mu(z)\log \cosh\Big(\sqrt{\frac{2}{\alpha^{(s)}}\mathcal{P}_s\Big(\mathbf{\Delta}\mathbf{q}_{trial}\Big)}z+h^{(s)}\Big),
\ee
and
$$z\sim\mathcal{N}(0,1).$$

Setting $K = 2$, $m_1 = 0$, $m_2 = 1$ and $\mathbf{q}_1 = \mathbf{q}_{trial}$  in Theorem \ref{thrsb} we obtain the following:

\begin{prop}\label{RST}
The following sum rule holds
\be\label{sumrule}
p_N=p_{RS}(\mathbf{q}_{trial})-\frac{1}{2}\int^1_0\mathbb{E}\Omega_{N,t}\Big((\mathbf{q}_N-\mathbf{q}_{trial}),\mathbf{\Delta} (\mathbf{q}_N-\mathbf{q}_{trial})\Big).
\ee
Moreover, if the matrix $\mathbf{\Delta}$ is positive semi-definite, then the following bound holds
\be\label{bound}
p_N\leq p_{RS}(\mathbf{q}_{trial}),
\ee
whose optimization gives
\be\label{boundopti}
p_N\leq \inf_{\mathbf{q}_{trial}} p_{RS}(\mathbf{q}_{trial}).
\ee
\end{prop}

\noindent The optimization of (\ref{boundopti}) on $\mathbf{q}_{trial}$,  gives a system of $S$ coupled self consistent equations, i.e. $\forall p \in\mathcal{S}$
\be\label{selfa}
\sum_{s\in\mathcal{S}}\boldsymbol{\Delta}_{ps}\left[\int d\mu(z)\tanh^2\Big(\sqrt{\frac{2}{\alpha^{(s)}}\mathcal{P}_s(\boldsymbol{\Delta} \mathbf{q}_{trial})}z\Big)-q^{(s)}\right]=0,
\ee
This system admits a unique solution as soon as $\det(\boldsymbol{\Delta})\neq 0$, thus
whenever $\det(\boldsymbol{\Delta})>0$, $p_{RS}(\mathbf{q}_{trial})$ has a minimum in $\mathbf{q}_{trial}=\mathbf{\bar{q}}$ satisfying $\forall s\in\mathcal{S}$

\be\label{selfcon2a}
\bar{q}^{(s)}=\int d\mu(z)\tanh^2\Big(\sqrt{\frac{2}{\alpha^{(s)}}\mathcal{P}_s(\boldsymbol{\Delta} \mathbf{\bar{q}})}z\Big)=\left\langle \mathcal{P}_s (\mathbf{q}) \right\rangle_{t=0}
\ee
\noindent The last equalities can be easily checked thanks to the factorizability of the one-body problem at $t=0$.
In other words, the value of $\mathbf{q}_{trial}$ minimizing the overlap' s fluctuations of the original model (at $t=1$) is just the overlap's mean of the interpolating one-body trial at $t=0$.

Let show now how the replica symmetric bound violate the entropy positivity at low temperatures.
Mirroring the historical scenario for mono-partite spin-glasses, we can easily check that the replica symmetric expression for the pressure $(\ref{boundopti})$ is not the exact solution of the model in the low temperature region by studying the behavior of the entropy. We can define it as  the {\em non-negative} quantity
\be
s(\boldsymbol{\Delta})=\lim_{N\to\infty}s_N(\boldsymbol{\Delta})=-\frac 1 N \mathbb{E}\sum_{\sigma}G_N(\sigma,\boldsymbol{\Delta})
\log(G_N(\sigma,\boldsymbol{\Delta})),
\ee
where $G_N(\sigma,\boldsymbol{\Delta})=Z_N^{-1}(\boldsymbol{\Delta})e^{-H_N(\sigma,\boldsymbol{\Delta})}$ is the Boltzmann measure. Notice that, unlike before, we have write explicitly the dependance on the matrix $\boldsymbol{\Delta}$ . Since $s_N(\boldsymbol{\Delta})=p_N(\boldsymbol{\Delta})-\frac 1 N \langle H(\sigma)\rangle_N $,  we can write
\be
s(\boldsymbol{\Delta})=p(\boldsymbol{\Delta})- \frac d {d\lambda} p(\lambda \boldsymbol{\Delta})|_{\lambda=1}.
\ee
Now we can define $s_{RS}(\boldsymbol{\Delta})=p_{RS}-\frac d {d\lambda} p_{RS}(\lambda \boldsymbol{\Delta})|_{\lambda=1}$. We can easily show that if the amplitude of the covariances is large enough, $s_{RS}(\boldsymbol{\Delta})$ is strictly negative. Indeed, we have the following
\begin{prop}
In the regime of large covariances (low temperatures), the RS-entropy is strictly negative, i.e.
$$
\lim_{\beta\to+\infty}s_{RS}(\beta \boldsymbol{\Delta})<0,
$$
for any choice of $\boldsymbol{\Delta}$ with $(\det (\boldsymbol{\Delta})>0)$ and $\boldsymbol{\alpha}$, where $\beta\in\mathbb{R}^{+}$ plays the role of the inverse temperature.
\end{prop}
\begin{proof} Using its definition
$$
s_{RS}(\beta \boldsymbol{\Delta})=p_{RS}(\beta\boldsymbol{\Delta},\bar{\mathbf{q}})-\frac{\partial}{\partial\lambda} p_{RS}(\lambda\beta \boldsymbol{\Delta},\bar{\mathbf{q}})|_{\lambda=1}.
$$
We note that, using  (\ref{selfcon2a}), in the limit $\beta\to +\infty$, the optimized order parameters $\bar{\mathbf{q}}\to \textbf{1}$. Explicating the derivative it is easy to see that
$$
\lim_{\beta\to+\infty}s_{RS}(\beta \boldsymbol{\Delta})=\lim_{\beta\to+\infty}-\frac{ \beta^2} 2\left((\boldsymbol{1}-\bar{\mathbf{q}}),\boldsymbol{\Delta} (\boldsymbol{1}-\bar{\mathbf{q}}) \right) \leq 0
$$
Finally we can state that the limit is strictly negative, using again (\ref{selfcon2a}) and noting that
\be
\begin{split}
\lim_{\beta\to+\infty}\beta (1- \bar{q}_{(s)})&=\lim_{\beta\to+\infty}\beta \int d\mu(z)\Big(1-\tanh^2\Big(\beta
\sqrt{\frac{2}{\alpha^{(s)}}\mathcal{P}_s\Big(\mathbf{\Delta}\mathbf{\bar{q}}\Big)}z\Big)\Big)\nonumber \\
&=\frac{\int d\mu(z) |z|}{\sqrt{\frac{2}{\alpha^{(s)}}\mathcal{P}_s\Big(\mathbf{\Delta}\mathbf{1}\Big)}}>0 .\nonumber
\end{split}
\ee
\end{proof}
The existence of a negative RS-entropy regime is a clear signal that the model is not always replica symmetric (certainly it is RS inside the annealed region defined in Theorem $\ref{annealed region}$) but there exists a region in which the pressure $p(\boldsymbol{\Delta})$ is strictly lower than its RS bound $p_{RS}(\boldsymbol{\Delta})$.

\section{The Broken Replica Symmetry bound}

In this section we give a proof of Theorem \ref{thrsb}.
It is enough to show that (\ref{sumrulersb}) holds, then we have straightforward  conclusions. The strategy is to apply the RSB interpolation scheme introduced in Appendix \ref{notation}.\\
We define the interpolating Hamiltonian as
\be\label{hamiintersb}
H_{N}(\sigma,t):=\sqrt{t}H_{N}(\sigma)+\sqrt{1-t}\sum^K_{l=1}H^l_{N}(\sigma,\mathbf{q}_{l}),
\ee
with
\be\label{hamitrialrsb}
H^l_{N}(\sigma,\mathbf{q}_{l}):=\sum_{s\in\mathcal{S}}H^{l,(s)}_{N}(\sigma^{(s)},\mathbf{q}_l),
\ee
where $H_{N}(\sigma)$ is the original Hamiltonian and, for each $l$, $H^{l,(s)}_{N}(\sigma^{(s)},\mathbf{q}_{l})$ are two independent one-body interaction Hamiltonian, defined as
\be\label{hami2rsb}
H^{l,(s)}_{N}(\sigma^{(s)},\mathbf{q}_{l}):=-\sqrt{2} \sqrt{\mathcal{P}_s\Big(\mathbf{\Delta}(\mathbf{q}_{l}-\mathbf{q}_{l-1})\Big)}\frac{1}{\sqrt{\alpha^{(s)}}}\sum_{i\in\Lambda^{(s)}_N} J^{l,(s)}_{i}\sigma^{(s)}_i
\ee
where the $J$'s are \textsl{Gaussian i.i.d.} r.v., independent of the other r.v., such that for every $l,s$ and $i$ we have that

\be\label{meanRSB}
\mathbb{E}(J^{l,(s)}_{i})=0
\ee
and
\be\label{variaRSB}
\mathbb{E}(J^{l,(s)}_{i}J^{l',(s')}_{i'})=\delta_{ll'}\delta_{ss'}\delta_{ii'}.
\ee
After simple computations, we get
$$\mathbb{E}\Big(H^{l,(s)}_{N}(\sigma^{(s)},\mathbf{q}_{l})H^{l',(s')}_{N}(\tau^{(s')},\mathbf{q}_{l})\Big)=\delta_{ll'}\delta_{ss'}
2N\mathcal{P}_s\Big(\mathbf{\Delta}(\mathbf{q}_{l}-\mathbf{q}_{l-1})\Big)\mathcal{P}_s\Big(\mathbf{q}_N\Big)$$
and then by (\ref{algebra}) the covariance matrix of the trial Hamiltonian becomes
$$\mathbb{E}\Big(H^l_{N}(\sigma,\mathbf{q}_{l})H^{l'}_{N}(\tau,\mathbf{q}_{l'})\Big)=
\delta_{ll'}2N\Big((\mathbf{q}_{l}-\mathbf{q}_{l-1}),\mathbf{\Delta}\mathbf{q}_{N}\Big).$$

Keeping in mind Proposition \ref{interrsb2}, we introduce the RSB interpolation scheme with the following identifications:
$$i\rightarrow\sigma, \ \ a_i\rightarrow a_N(\sigma,\mathbf{h}), \ \  U_i\rightarrow H_{N}(\sigma), \  \
B_i^{l}\rightarrow H^{l}_{N}(\sigma,\mathbf{q}_{l})$$
and we define the interpolating pressure as
\be\label{qpresinter}
p_N(t):=\frac{1}{N}\mathbb{E}\log Z_{0,N}(t).
\ee
\\
It is easy to check that the boundary values of $p_N (t)$ are
\begin{eqnarray}\label{boundary1}
p_N(1)&=&p_N,\\ \label{boundary0}
p_N(0)&=&\log2+\sum_{s\in\mathcal{S}}\alpha^{(s)}f^{(s)}(0,h^{(s)}),
\end{eqnarray}
where $f^{(s)}(u^{(s)},h^{(s)})$ is the solution of the Parisi's PDE \eqref{antipara}.
\\
In order to apply the interpolation argument we have to compute the $t$-derivative of the interpolating pressure. A simple application of Proposition \ref{interrsb2} leads to the following

\be\label{derivativersb}
\frac{\partial}{\partial t}p_N(t)=-\frac{1}{2}\Big(\mathbf{1},\mathbf{\Delta}\mathbf{1}\Big)-
\frac{1}{2}\sum^K_{l=0}(m_{l+1}-m_l)\Big\langle\Big(\mathbf{q}_N\mathbf,{\Delta}\mathbf{q}_{N}\Big)-2\Big(\mathbf{q}_N\mathbf,{\Delta}\mathbf{q}_{l}\Big)\Big\rangle_{N,l,t}.
\ee
To complete the proof of the Theorem, we need the following equivalence
\begin{prop}\label{intrepre}
The following representation holds
$$-\frac{1}{2}\Big(\mathbf{1},\mathbf{\Delta}\mathbf{1}\Big)+
\frac{1}{2}\sum^K_{l=0}(m_{l+1}-m_l)\Big(\mathbf{q}_l,\mathbf{\Delta}\mathbf{q}_{l}\Big)=
-\frac{1}{2}\int_{\widetilde{\Gamma}} d\mathbf{u}\ \ x(\mathbf{u})\ \ \nabla_{\mathbf{u}}\Big(\mathbf{u},\mathbf{\Delta}\mathbf{u}\Big)\cdot d\mathbf{u}.
$$
\end{prop}
\begin{proof} We can use the explicit definition of $x(\mathbf{u})$ given in (\ref{orderpara}) to check that

$$-\frac{1}{2}\int_{\widetilde{\Gamma}} \ \ x(\mathbf{u})\ \ \nabla_{\mathbf{u}}\Big(\mathbf{u},\mathbf{\Delta}\mathbf{u}\Big)\cdot d\mathbf{u}
=-\frac{1}{2}\sum^{K}_{l=0}(m_{l+1}-m_l)\int_{\Gamma_l}\ \ \nabla_{\mathbf{u}}\Big(\mathbf{u},\mathbf{\Delta}\mathbf{u}\Big)\cdot d\mathbf{u}
$$
where ${\Gamma_l}$ is the result of the action of the $\theta$'s
on the path ${\widetilde{\Gamma}}$, that is his component  between the points $\mathbf{q}_l $ and $\mathbf{1}$. By the Gradient's Theorem, the integral is path independent and is equal to the increment of the potential function, that is the desired result.
\end{proof}
Finally, combining (\ref{boundary1}), (\ref{boundary0}), (\ref{derivativersb}) and Proposition \ref{intrepre}, the proof of (\ref{sumrulersb}) is a simple application of the fundamental theorem of calculus.

\appendix

\section{Notation and technical tools}\label{notation}

Much of the recent progresses in the study of mean field spin glass models is based on methods and arguments
introduced by Guerra in a series of works (see e.g. \cite{broken,guerra2,GT2,limterm,dibiasio}), constituting  the so called {\em interpolation method}.
Beyond the original works, the interested reader can found a detailed and complete exposition with several applications of this method in \cite{talbook}, while in \cite{barra1} these techniques are shown at work on the simpler Curie-Weiss model.
In order to present a self-consistent exposition, hereafter we outline briefly the basic ideas.

Let $N$ be an integer, and for $i\in I=\{1, . . . ,N\}$,
let $U_i$ and $\widetilde{U}_i$ be two families of centered
Gaussian random variables, independent each other, uniquely determined by the respective covariance matrices
$\mathbb{E}(U_iU_j) = C_{ij}$ and $\mathbb{E}(\widetilde{U}_i\widetilde{U}_j) = \widetilde{C}_{ij}$.  We treat the set of indices $i$  as configuration space for
some statistical mechanics system. We define the Hamiltonian interpolating function as the following random variable
$$H_i(t):= \sqrt{t}U_i+\sqrt{1-t}\tilde{U}_i,$$
where $t\in[0,1]$ is the real parameter used for interpolation. 
%
%

Let us introduce  the multi-species framework. Suppose that the system is composed by a finite number $S$ of species indexed by $s\in\mathcal{S}$, then  $|\mathcal{S}|=S$. We assume that:\smallskip

- the configuration space is decomposed in a disjoint union:
 $$I=\bigcup_{s\in\mathcal{S}}I^{(s)},$$\smallskip

- the $U$'s are also decomposed in the following way:

\be\label{decompu}
U_i=\sum_{s,p \in \mathcal{S}}U_i^{(sp)},
\ee
\noindent
where $U_i^{(sp)}$ is a family of gaussian r.v. such that the covariance matrix is of the form
\be
\label{covaspecies}
\mathbb{E}(U_i^{(sp)}U_j^{(s'p')})=\Delta^2_{sp}\delta_{ss'}\delta_{pp'}C_{ij}^{(s)}C_{ij}^{(p)}
\ee
where $C_{ij}^{(s)}$ is a covariance matrix defined on $I^{(s)}\times I^{(s)}$.

Notice that the covariance matrix defined in (\ref{covaspecies}) is the Schur-Hadamard product of the $C_{ij}^{(s)}$ and then is positive definite. The family of positive parameters $(\Delta^2_{sp})_{s,p\in\mathcal{S}}$ tunes the interactions between the various species.\\
For a fixed couple $(i,j)$ we can think at each $C_{ij}^{(s)}$ as a component of a vector in the space $\mathbb{R}^{S}$ and then, thanks to (\ref{decompu}) and (\ref{covaspecies}), the covariance matrix of the entire system can be rewritten, with a slightly abuse of notation, as a quadratic form in $\mathbb{R}^{S}$, namely as
\be
\label{covasystemU}
C_{ij}=\mathbb{E}(U_i U_j)=\sum_{s,p \in \mathcal{S}}C_{ij}^{(s)}\Delta^2_{sp}C_{ij}^{(p)}=\Big(\mathbf{C},\mathbf{\Delta}\mathbf{C}\Big),
\ee
where $\mathbf{C}:=(C_{ij}^{(s)})_{s\in\mathcal{S}}$ is a vector in $\mathbb{R}^{S}$ and
$\mathbf{\Delta}$ is the real symmetric matrix defined by the entries
$$\mathbf{\Delta}:=(\Delta^2_{sp})_{s,p \in \mathcal{S}}.$$
Suppose for simplicity that $C_{ii}^{(s)}=\sqrt{c} $ for some $c\in \mathbb{R}^+$ for each $i\in I,s\in\mathcal{S}$, that is
\be\label{covadiagU}
C_{ii}=c(\mathbf{1},\mathbf{\Delta}\mathbf{1}),
\ee
where
$$\mathbf{1}:=(1)_{s\in\mathcal{S}}.$$
We make the same assumptions for the $\widetilde{U}$'s.
An application of integration by part formula gives the following
\begin{prop}\label{gauinter}
Consider the functional
\be\label{funinter0}
\varphi(t):=\mathbb{E}\log Z(t)
\ee
then for its $t$-derivative the following holds
\be\label{inter1}
\varphi'(t)=\frac{1}{2}(c-\widetilde{c}) (\mathbf{1},\mathbf{\Delta}\mathbf{1}) -\frac{1}{2}{\langle (\mathbf{C},\mathbf{\Delta}\mathbf{C})-
(\mathbf{\widetilde{C}},\mathbf{\Delta}\mathbf{\widetilde{C}})\rangle}_t.
\ee
\end{prop}

where $\langle \ \ \rangle_t$ is the quenched measure (\ref{qeqstate}) induced by $H_i(t)$.
\\
In order to separate the contribution of the various species, let us introduce the operator $\mathcal{P}_s$ as the canonical projector in $\mathbb{R}^{S}$.\\
For any $s\in\mathcal{S}$ and for any vector $\mathbf{u}=(u^{(s)})_{s\in\mathcal{S}}$ in $\mathbb{R}^{S}$, we have that
\be
\label{projection}
\mathcal{P}_s\Big(\mathbf{u}\Big):=u^{(s)}.
\ee
Clearly, for two vectors $\mathbf{u},\mathbf{v}$, the following relation holds
\be\label{algebra}
\Big(\mathbf{u},\mathbf{\Delta}\mathbf{v}\Big)=\sum_{s\in\mathcal{S}}\mathcal{P}_s\Big(\mathbf{u}\Big)\mathcal{P}_s\Big(\mathbf{\Delta}\mathbf{v}\Big)=
\sum_{s\in\mathcal{S}}\mathcal{P}_s\Big(\mathbf{\Delta}\mathbf{u}\Big)\mathcal{P}_s\Big(\mathbf{v}\Big).
\ee
We discuss now briefly the RSB Guerra's interpolation for multipartite systems in the same setting of the original work \cite{broken}.
\\

Let $K$ be an integer and consider an arbitrary sequence of points $\Gamma:=(\mathbf{q}_l)_{l=1,\ldots, K}\in [0,1]^{S}$. For each triple $(l, i, s)$ with $l = 1, 2, . . . ,K$, $ i\in I$, $s\in\mathcal{S}$, let us introduce the
family of centered Gaussian random variables $B^{l,(s)}_i$ independent from the $U_i$
and uniquely defined through the covariances
\be
\label{covaRSB}
\mathbb{E}(B^{l,(s)}_iB^{l',(s')}_j)=
\delta_{ss'}\delta_{ll'}\mathcal{P}_s\Big(\mathbf{\Delta}\mathbf{u}_{l}{(\Gamma)}\Big)\mathcal{P}_s\Big(\mathbf{\widetilde{C}}_{l}\Big)
\ee
where, for each value of $l$, the component of the vector  $\mathbf{\widetilde{C}}_{l}=(\widetilde{C}^{(s)}_{l,ij})_{s\in\mathcal{S}}$,
are covariance matrix defined on $I^{(s)}\times I^{(s)}$ and $\mathbf{u}_{l}{(\Gamma)}$ is an arbitrary  vector in $\mathbb{R}^{S}$  which depends on the choice of the sequence $\Gamma$ .
\\
Notice that (\ref{covaRSB}) implies independence between  two different $l^{(s)},l'^{(s)}$ levels of symmetry breaking of each $s$-species.
For each $l = 1, 2, . . . ,K$ and $ i\in I$, we can define the following family of random variables
$$ B^{l}_i:=\sum_{s\in\mathcal{S}}B^{l,(s)}_i$$
then by (\ref{algebra}) we have that
\be
\label{familycova}
\mathbb{E}(B^{l}_iB^{l'}_j)=\delta_{ll'}\Big(\mathbf{u}_{l}{(\Gamma)},\mathbf{\Delta}\mathbf{\widetilde{C}}_{l}\Big).
\ee
Suppose for simplicity that $\widetilde{C}_{l,ii}^{(s)}=1$ for each $l,i,s$, that is
$$\mathbb{E}(B^{l}_iB^{l'}_i)=\delta_{ll'}\Big(\mathbf{u}_{l}{(\Gamma)},\mathbf{\Delta}\mathbf{1}\Big).$$
Let us introduce the following notations for the average with respect to $B^l_i,U_i$,
\begin{eqnarray}
\mathbb{E}_l(\cdot)&=&\int \prod_{i}d\mu(B^l_i)(\cdot)\ \ \ \ \ \ \forall l=1,...,K, \\
d\mu(B^l_i)&=&\prod_{s\in\mathcal{A}}d\mu(B^{l,(s)}_i), \\
\mathbb{E}_0(\cdot)&=&\int \prod_{i}d\mu(U_i)(\cdot), \\
\mathbb{E}(\cdot)&=&\mathbb{E}_0\mathbb{E}_1...\mathbb{E}_K(\cdot).
\end{eqnarray}
Consider a non-decreasing sequence of non negative real numbers $(m_0, m_1, ..., m_K, m_{K+1})$ with $m_0=0, \ m_{K+1}=1$ and define recursively the following the random variables
\begin{eqnarray}
Z_K(t)&:=&\sum_{i}\omega_i\exp{(\sqrt{t}U_i+\sqrt{1-t}\sum^K_{l=1}B^l_i)}, \\
Z^{m_l}_{l-1}&:=&\mathbb{E}_l(Z^{m_l}_{l})
\end{eqnarray}
\begin{prop}\label{interrsb2}
Consider the functional
\be
\label{funinterrsb}
\varphi(t)=\mathbb{E}_0\log(Z_0(t)),
\ee
then for its $t$-derivative the following relation holds
\be \label{interrsb}
\varphi'(t)=\frac{1}{2} (\mathbf{1},\mathbf{\Delta}\mathbf{1})-\sum^K_{l=1}\Big(\mathbf{u}_{l}(\Gamma),\mathbf{\Delta}\mathbf{1}\Big)
-\frac{1}{2}\sum^K_{l=0}(m_{l+1}-m_l){\langle  \Big(\mathbf{C},\mathbf{\Delta}\mathbf{C}\Big)-\widehat{B}^l\rangle }_{l,t}\ee
where $\widehat{B}^0:=0,\ \widehat{B}^l:=\sum^l_{l'=1}\Big(\mathbf{u}_{l}{(\Gamma)},\mathbf{\Delta}\mathbf{\widetilde{C}}_{l'}\Big)$ and ${\langle \ \ \rangle}_{l,t}$ is a suitable deformed quenched measure.
\end{prop}

We notice that the previous proposition can be restated in the language of Ruelle Probability Cascades
\cite{PanchenkoTal}.

\section*{Acknowledgements}

Authors are grateful to MiUR trough the FIRB grants number
$RBFR08EKEV$ and $RBFR10N90W$ and PRIN grant number $2010HXAW77$ and to Sapienza University of Rome and Alma Mater Studiorum, Bologna University. AB is partially funded by GNFM (Gruppo Nazionale per la Fisica
Matematica) which is also acknowledged.



\end{document}